\documentclass[journal,twoside,final]{IEEEtran}
\normalsize
\usepackage{graphicx}
\usepackage{balance}
\usepackage{amsmath,amssymb,amsthm}
\usepackage{bm}
\usepackage[mathscr]{euscript}
\usepackage{bigints}
\usepackage[FIGTOPCAP]{subfigure}
\usepackage{cite}
\usepackage{xcolor}
\usepackage{multirow}
\usepackage{nicematrix}
\usepackage{mathrsfs}
\usepackage{amsfonts}
\usepackage{cuted}
\usepackage{psfrag}
\usepackage{tipa}
\usepackage{array}
\usepackage{soul}
\usepackage{setspace}
\usepackage[pdf]{epsfig}
\usepackage{booktabs}
\usepackage{tabularx}
\usepackage{makecell}
\usepackage{rotating}
\usepackage{colortbl}
\definecolor{rowalt}{RGB}{245,245,245}

\theoremstyle{plain}
\newtheorem{theorem}{Theorem}
\newtheorem{corollary}{Corollary}

\begin{document}

\title{Distance Distributions Between Nodes in Concentric Disk–Annulus or Sphere–Shell Regions}
\author{Nicholas~Vaiopoulos,~Alexander~Vavoulas,~Harilaos~G.~Sandalidis,~and~Konstantinos~K.~Delibasis
\IEEEcompsocitemizethanks{\IEEEcompsocthanksitem Nicholas Vaiopoulos,~Alexander~Vavoulas,~Harilaos~G.~Sandalidis,~and~Konstantinos~K.~Delibasis are with the Department of Computer Science and Biomedical
Informatics, University of Thessaly, Papasiopoulou 2-4, 35131, Lamia, Greece. e-mail:\{nvaio,vavoulas,sandalidis,kdelimpasis\}@dib.uth.gr.}
}
\maketitle

\begin{abstract}
This letter derives closed-form expressions for the probability density function of the distance between two nodes located in heterogeneous concentric geometries, namely a disk or sphere and a surrounding annulus or spherical shell. Two scenarios are considered: (i) both nodes are independently distributed in different regions, disk or sphere and annulus or shell, and (ii) one node is static in the outer region while the other follows the stationary distribution of the random waypoint model in the inner region. The resulting expressions provide a tractable analytical tool for performance evaluation in concentric wireless regions.
\end{abstract}


\maketitle
\section{Introduction}

\IEEEPARstart{I}{nternodal} distance statistics are often studied in bounded regions arising from guard zones, exclusion zones, protected infrastructure, or mobility constraints. These settings lead to two-dimensional (2-D) annular geometries or three-dimensional (3-D) spherical shells, which are commonly used in cellular systems, device-to-device (D2D) and Internet of Things (IoT) deployments, relay-assisted networks, and unmanned aerial vehicle (UAV) or non-terrestrial networks. In this context, distance statistics are a key tool for wireless system analysis and design, as node separation directly affects path loss, interference, outage probability, connectivity, and energy consumption \cite{J:Armeniakos}.

While classical results exist for disks and spheres \cite{B:Mathai, J:Moltchanov, J:Parry}, their extension to annular and spherical shell geometries under non-uniform steady-state distributions remains challenging due to multiple intersection regimes.
This letter fills this gap by deriving exact, closed-form, piecewise probability density functions (PDFs) for the distance between two independently distributed nodes in bounded regions. Specifically, results are obtained for a disk of radius $R_1$ and an annulus with an inner radius $R_1$ and an outer radius $R_2>R_1$ in 2-D, as well as for a sphere of radius $R_1$ and a spherical shell with radii $R_1$ and $R_2$ in 3-D. Two scenarios are examined: 
\begin{enumerate}
    \renewcommand{\labelenumi}{\roman{enumi})}
    \item Scenario $s_1$: Both nodes are static, with node~1 located within the disk or sphere of radius $R_1$ and node~2 within the concentric annular region or spherical shell defined by $R_1$ and $R_2$.
    
    \item Scenario $s_2$: Node~1 is mobile within the disk or sphere of radius $R_1$, while node~2 remains static within the same concentric annular region or spherical shell defined by $R_1$ and $R_2$.
\end{enumerate}
Static nodes are assumed to be uniformly distributed over their respective regions, while mobile nodes follow the stationary spatial distribution of the random waypoint (RWP) model \cite{J:Hyytia} with zero pause time, ignoring trajectory dynamics and temporal correlations.

\section{Model Assumptions}
The analysis builds on the framework of \cite{J:Vaiopoulos}, with the key distinction that the two nodes are located in different geometric regions: one within a disk or sphere and the other within an annulus or spherical shell. In contrast, \cite{J:Vaiopoulos} assumes that both nodes are confined to either a disk or a sphere. We consider three geometric regimes based on the ratio $R_2/R_1$, namely $R_2 \leq 2R_1$, $2R_1 < R_2 \leq 3R_1$, and $R_2 > 3R_1$. However, the first two regimes (Tables I and II) lead to identical expressions under all cases and can therefore be unified into a single regime $R_2 \leq 3R_1$.

In general, the analytical expression for the PDF of the internodal distance, $r$, is obtained as
\begin{equation}
  f_{\mathbf{r}}(r)=\int_{0}^{\infty}  f_{\mathbf{r}|\bm{\rho}}(r|\rho)f_{\bm{\rho}}(\rho)\mathrm{d}\rho,
  \label{CondProb} 
\end{equation}
where $f_{\mathbf{r}|\bm{\rho}}(r|\rho)$ denotes the conditional PDF of $r$, given that node 1 is located at a distance $\rho$ from the center. The term $f_{\bm{\rho}}(\rho)$ represents the PDF of the distance of node 1, corresponding to the RWP or the uniform distribution given for 2-D
{\small
\begin{align}
  f_{\bm{\rho}}^{\text{RWP}}(\rho) =\dfrac{4\rho}{R_{1}^{2}}(1-\dfrac{\rho^{2}}{R_{1}^{2}}),
\quad
 f_{\bm{\rho}}^{\text{U}}(\rho)  = \dfrac{2\rho}{R_{1}^{2}}.
  \label{UNrho1} 
\end{align}
}
and for 3-D
{\small
\begin{align}
  f_{\bm{\rho}}^{\text{RWP}}(\rho)\hspace{-2pt} =\hspace{-2pt} \dfrac{35}{72}\left(\dfrac{21 \rho^2}{R_{1}^3}\hspace{-2pt}-\hspace{-2pt}\dfrac{34 \rho^4}{R_{1}^5}\hspace{-2pt}+\hspace{-2pt}\dfrac{13 \rho^6}{R_{1}^7} \right),
\quad
 f_{\bm{\rho}}^{\text{U}}(\rho)\hspace{-2pt}=\hspace{-2pt}\dfrac{3{\rho}^2}{R_{1}^{3}}.
  \label{UNrho2} 
\end{align}
}
depending on the scenario under consideration \cite{J:Vaiopoulos}.
To compute $f_{\mathbf{r}|\bm{\rho}}(r|\rho)$, node~2 is assumed to lie on a circle or a spherical surface of radius $r$ centered at node~1. The conditional PDF is determined by the portion of this locus that falls within the annular region or spherical shell defined by radii $R_1$ and $R_2$. Depending on the relative values of $r$, $\rho$, $R_1$, and $R_2$, this intersection region may be fully contained, partially overlapped, or entirely outside the annular region or spherical shell.

Tables \ref{Table1}--\ref{Table3} specify the intervals of $r$ over which the conditional PDF takes distinct forms, together with the corresponding integration limits in \eqref{CondProb}. Overall, five closed-form expressions for $f_{\mathbf{r}|\bm{\rho}}(r|\rho)$ are derived, denoted by $g_i(\rho,r)$ for $i=0,\ldots,4$, as detailed in Appendix~A.
\begin{enumerate}
    \renewcommand{\labelenumi}{\roman{enumi})}
    \item The circle or spherical surface, illustrated by the dashed curve in Tables~\ref{Table1} -- \ref{Table3}, lies entirely outside the annular region or spherical shell defined by radii $R_1$ and $R_2$. In this case, $g_0(\rho,r)=0$.
    
    \item The circle or spherical surface partially intersects the annular region or spherical shell at the inner boundary $R_1$. In this case, $g_1(\rho,r)$ applies.
    
    \item The circle or spherical surface intersects the annular region or spherical shell, overlapping the inner boundary at $R_1$ and extending beyond the outer boundary at $R_2$. In this case, $g_2(\rho,r)$ applies.
    
    \item The circle or spherical surface lies entirely within the annular region or spherical shell. In this case, $g_3(r)$ applies.
    
    \item The circle or spherical surface partially intersects the annular region or spherical shell and extends beyond the outer boundary at $R_2$, without intersecting the inner boundary. In this case, $g_4(\rho,r)$ applies.
\end{enumerate}

\begin{table*}[t]
\centering
\caption{Regime 1: $R_2\leq2R_1$.}
\label{Table1} 
\resizebox{0.95\textwidth}{!}{%
\renewcommand{\arraystretch}{1.0}
\begin{NiceTabular}{c c c c c c}[hvlines]  
    \hline 
    \multicolumn{2}{c} {\small \textbf{Interval I}: $0\leq r\leq R_{2}-R_{1}$} 
    & \multicolumn{3}{c} {\small \textbf{Interval II}: $R_{2}-R_{1}<r\leq R_{1}$}  \\ 

    \small $0\leq\rho \leq R_{1}-r$ 
     & \small $R_{1}-r<\rho\leq R_{1}$  & \small $0\leq\rho\leq R_{1}-r$  
    & \small $R_{1}-r<\rho\leq R_{2}-r$  & \small $R_{2}-r<\rho\leq R_{1}$ \\   

    \includegraphics[keepaspectratio,width=1.2in]{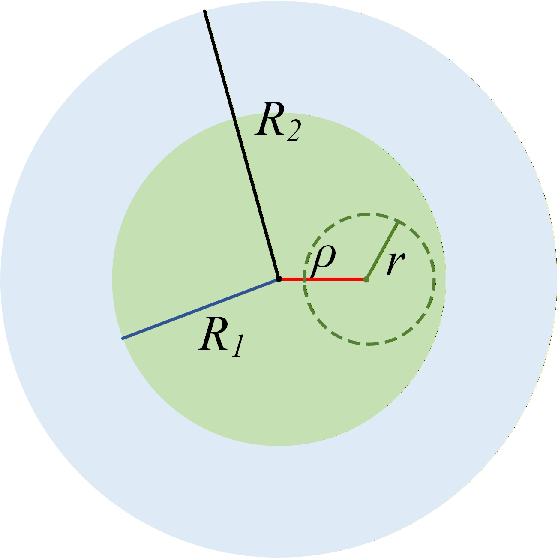}  
    & \includegraphics[keepaspectratio,width=1.2in]{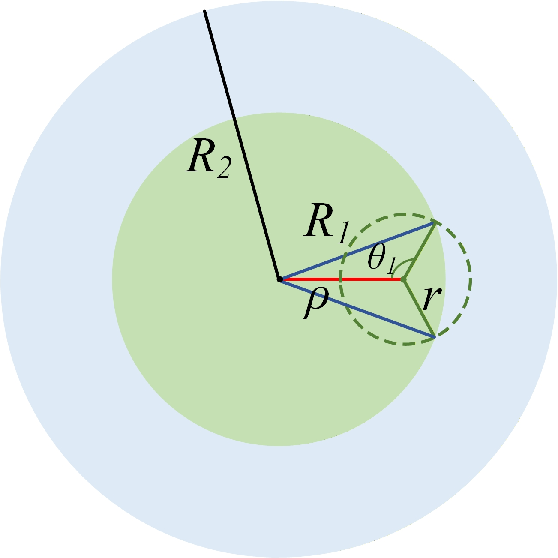}   
    &  \includegraphics[keepaspectratio,width=1.2in]{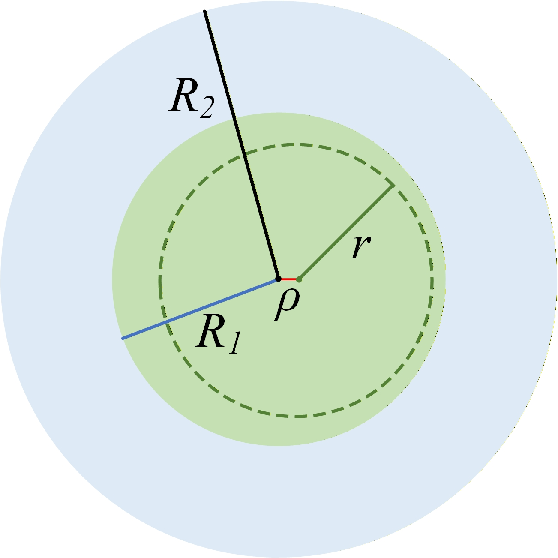}  
    & \includegraphics[keepaspectratio,width=1.2in]{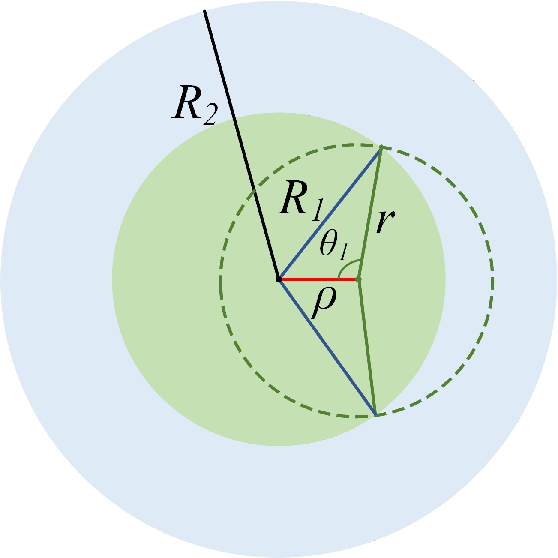}   
    &  \includegraphics[keepaspectratio,width=1.2in]{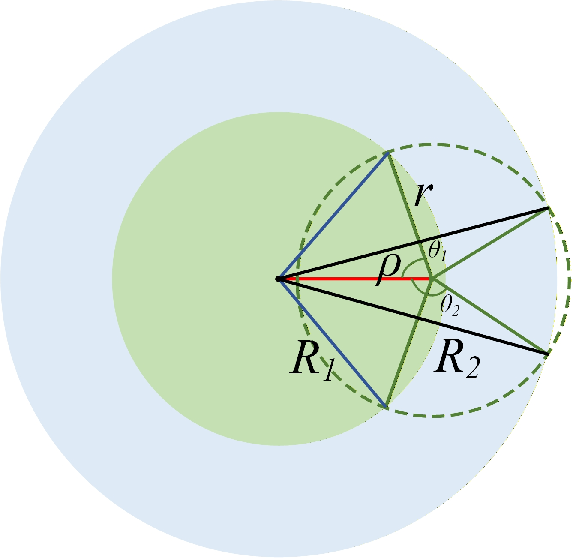}  
        &  \\
         \multicolumn{2}{c} {\small $f_{\mathbf{r}}(r)=\int_{R_{1}-r}^{R_{1}} g_{1}(\rho,r)f_{\rho}(\rho) \,\mathrm{d}\rho$} &
         \multicolumn{3}{c} {\small $f_{\mathbf{r}}(r)=\int_{R_{1}-r}^{R_{2}-r} g_{1}(\rho,r)f_{\rho}(\rho) \,\mathrm{d}\rho+\int_{R_{2}-r}^{R_{1}} g_{2}(\rho,r)f_{\rho}(\rho) \,\mathrm{d}\rho$}  \\
\hline \hline

    \multicolumn{3}{c} {\small \textbf{Interval III}: $R_{1}<r\leq\frac{R_{1}+R_{2}}{2}$} 
    & \multicolumn{3}{c} {\small \textbf{Interval IV}: $\frac{R_{1}+R_{2}}{2}<r\leq R_{2}$}  \\ 
 \small $0\leq\rho\leq r-R_{1}$ 
     & \small $r-R_{1}<\rho\leq R_{2}-r$  & \small $R_{2}-r<\rho\leq R_{1}$  
    & \small $0\leq\rho\leq R_{2}-r$  & \small $R_{2}-r<\rho\leq r-R_{1}$ & \small $r-R_{1}<\rho\leq R_{1}$ \\    
    
 \includegraphics[keepaspectratio,width=1.3in]{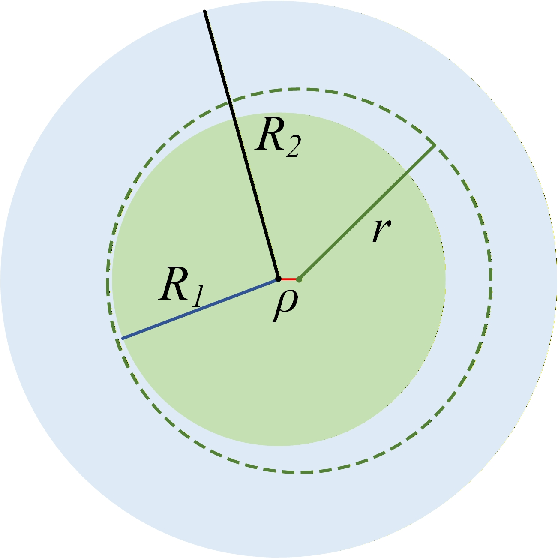}  
    & \includegraphics[keepaspectratio,width=1.3in]{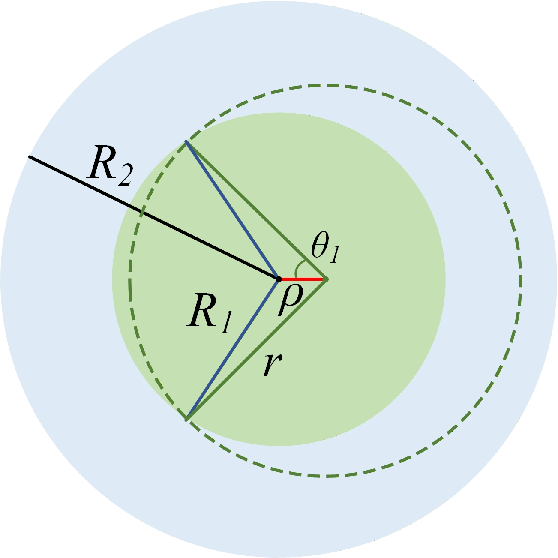}   
    &  \includegraphics[keepaspectratio,width=1.3in]{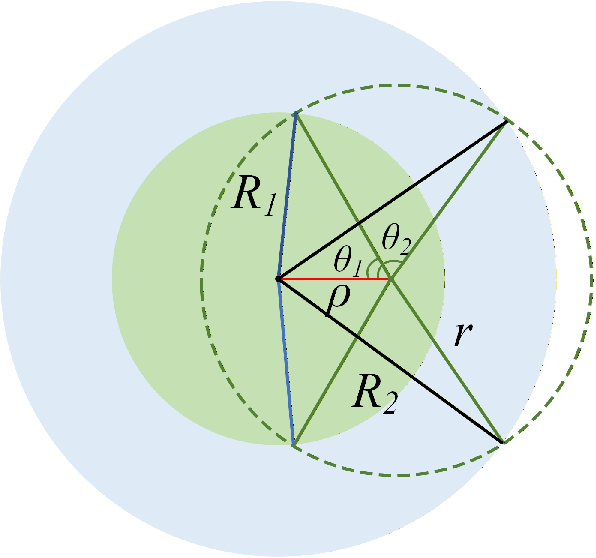}  
    & \includegraphics[keepaspectratio,width=1.3in]{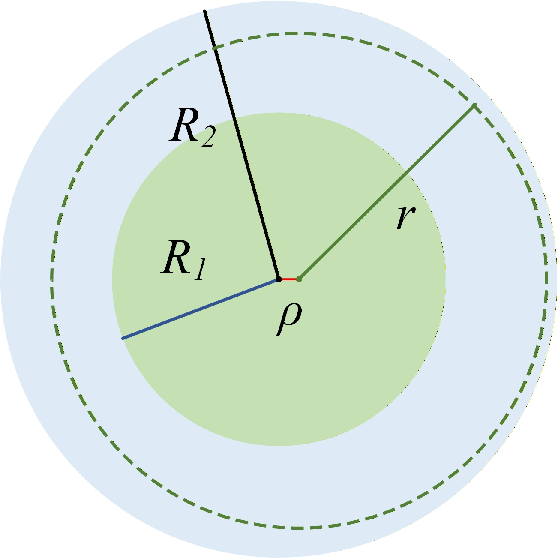}   
    &  \includegraphics[keepaspectratio,width=1.3in]{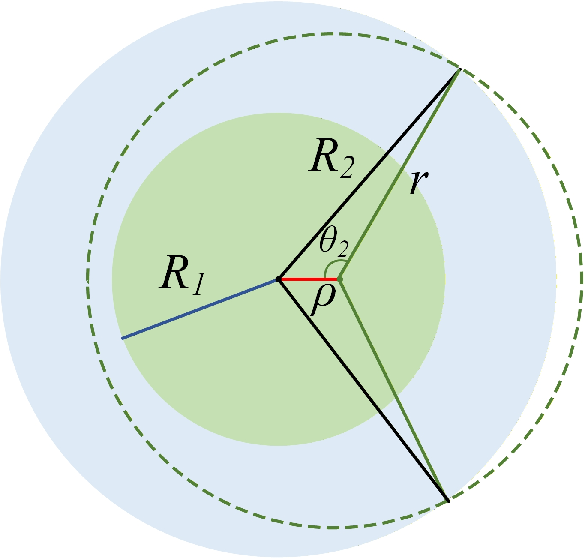}  
        &  \includegraphics[keepaspectratio,width=1.3in]{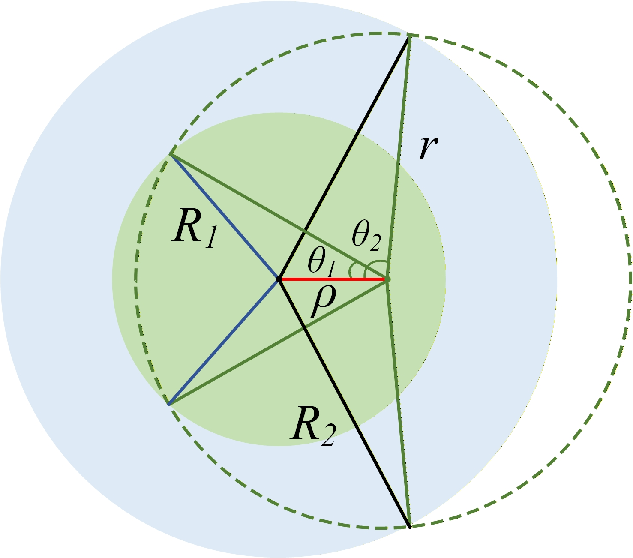}\\
         \multicolumn{3}{c} {\small $f_{\mathbf{r}}(r)=\hspace{-8pt}\int\limits_{0}^{r-R_{1}}\hspace{-8pt} g_{3}(r)f_{\rho}(\rho) \,\mathrm{d}\rho+\hspace{-8pt}\int\limits_{r-R_{1}}^{R_{2}-r}\hspace{-8pt} g_{1}(\rho,r)f_{\rho}(\rho) \,\mathrm{d}\rho+\hspace{-8pt}\int\limits_{R_{2}-r}^{R_{1}}\hspace{-8pt} g_{2}(\rho,r)f_{\rho}(\rho) \,\mathrm{d}\rho$} &  \multicolumn{3}{c} {\small $f_{\mathbf{r}}(r)=\hspace{-8pt}\int\limits_{0}^{R_{2}-r}\hspace{-8pt} g_{3}(r)f_{\rho}(\rho) \,\mathrm{d}\rho+\hspace{-8pt}\int\limits_{R_{2}-r}^{r-R_{1}}\hspace{-8pt} g_{4}(\rho,r)f_{\rho}(\rho) \,\mathrm{d}\rho+\hspace{-8pt}\int\limits_{r-R_{1}}^{R_{1}}\hspace{-6pt} g_{2}(\rho,r)f_{\rho}(\rho) \,\mathrm{d}\rho$}\\
\hline \hline

    \multicolumn{3}{c} {\small \textbf{Interval V}: $R_{2}<r\leq2R_{1}$} 
    & \multicolumn{2}{c} {\small \textbf{Interval VI}: $2R_{1}<r\leq R_{1}+R_{2}$}  \\ 
 \small $0\leq\rho\leq r-R_{2}$ 
     & \small $r-R_{2}<\rho\leq r-R_{1}$  & \small $r-R_{1}<\rho\leq R_{1}$  
    & \small $0\leq\rho\leq r-R_{2}$  & \small $r-R_{2}<\rho\leq R_{1}$  \\    
    
 \includegraphics[keepaspectratio,width=1.2in]{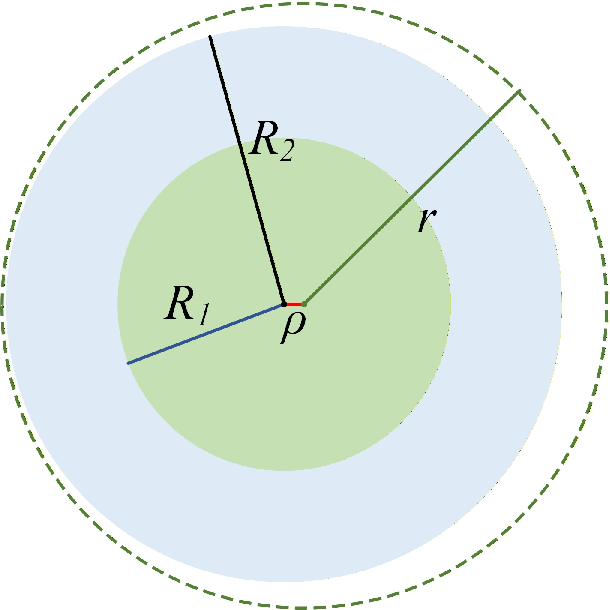}  
    & \includegraphics[keepaspectratio,width=1.2in]{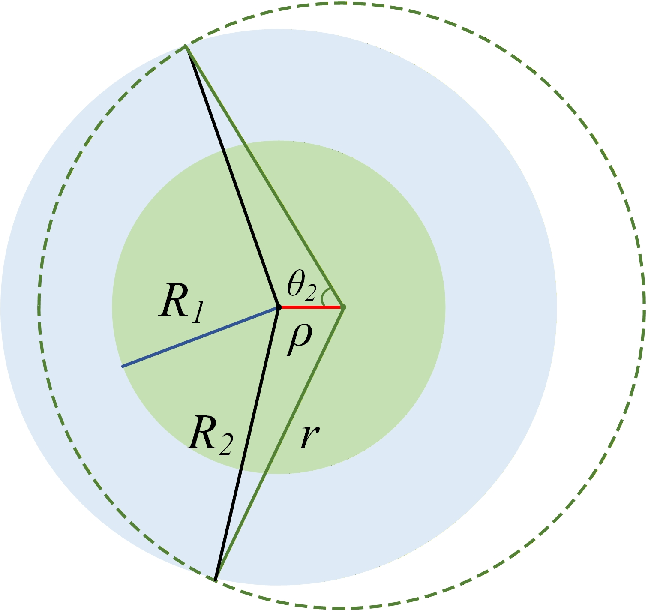}   
    &  \includegraphics[keepaspectratio,width=1.4in]{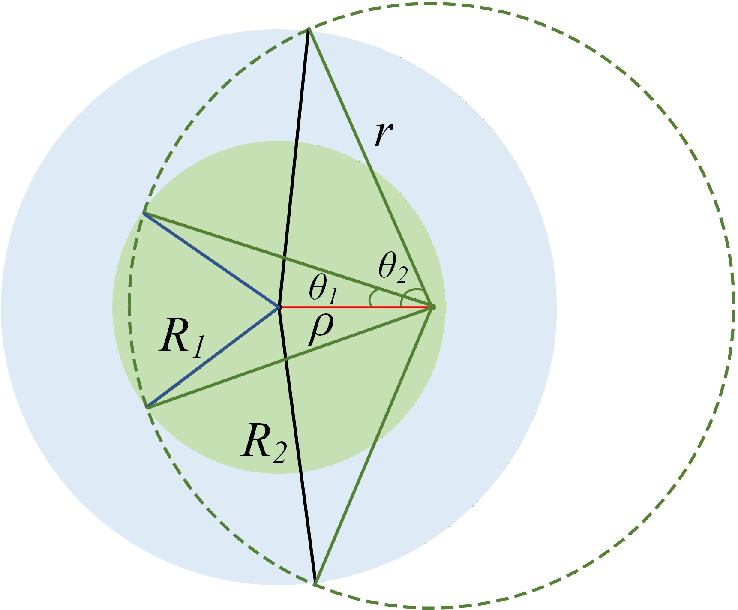} 
    & \includegraphics[keepaspectratio,width=1.2in]{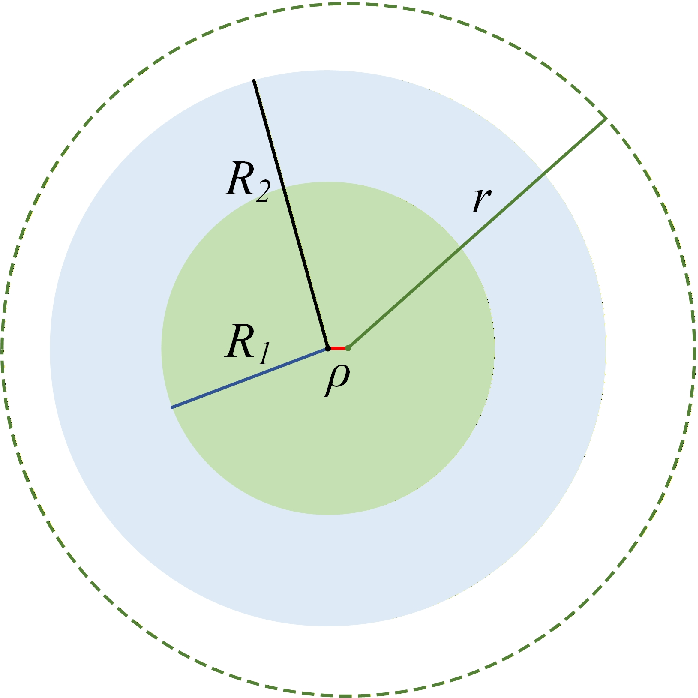}   
    &  \includegraphics[keepaspectratio,width=1.2in]{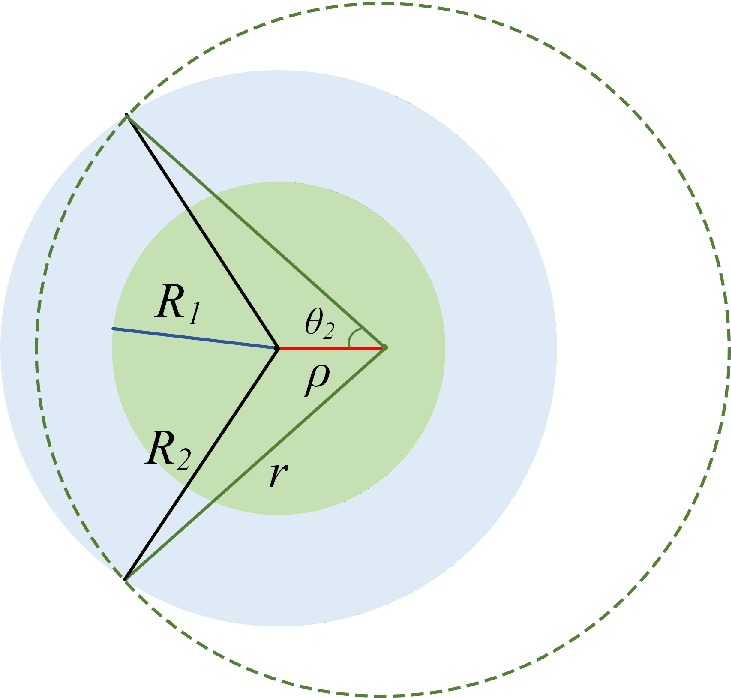}  
        &  \\
\multicolumn{3}{c} {\small $f_{\mathbf{r}}(r)=\int_{r-R_{2}}^{r-R_{1}} g_{4}(\rho,r)f_{\rho}(\rho) \,\mathrm{d}\rho+\int_{r-R_{1}}^{R_{1}} g_{2}(\rho,r)f_{\rho}(\rho) \,\mathrm{d}\rho$} &  \multicolumn{2}{c} {\small $f_{\mathbf{r}}(r)=\int_{r-R_{2}}^{R_{1}} g_{4}(\rho,r)f_{\rho}(\rho) \,\mathrm{d}\rho$}\\
\hline
\end{NiceTabular}
}
\end{table*}

\begin{table*}[t]
\centering
\caption{Regime 2: $2R_1<R_2\leq3R_1$.}
\label{Table2} 
\resizebox{0.9\textwidth}{!}{%
\renewcommand{\arraystretch}{1.6}
\begin{NiceTabular}{c c c c c c c}[hvlines]  
\hline 

\multicolumn{2}{c}{\large \textbf{Interval I}: $0\leq r\leq R_{1}$} 
& \multicolumn{2}{c}{\large \textbf{Interval II}: $R_{1}<r\leq R_{2}-R_{1}$} 
& \multicolumn{3}{c}{\large \textbf{Interval III}: $R_{2}-R_{1}<r\leq\frac{R_{1}+R_{2}}{2}$} \\ 

\large $0\leq\rho\leq R_{1}-r$ 
& \large $R_{1}-r<\rho\leq R_{1}$  
& \large $0\leq\rho\leq r-R_{1}$  
& \large $r-R_{1}<\rho\leq R_{1}$  
& \large $0\leq\rho\leq r-R_{1}$ 
& \large $r-R_{1}<\rho\leq R_{2}-r$  
& \large $R_{2}-r<\rho\leq R_{1}$ \\    

\includegraphics[width=1.3in]{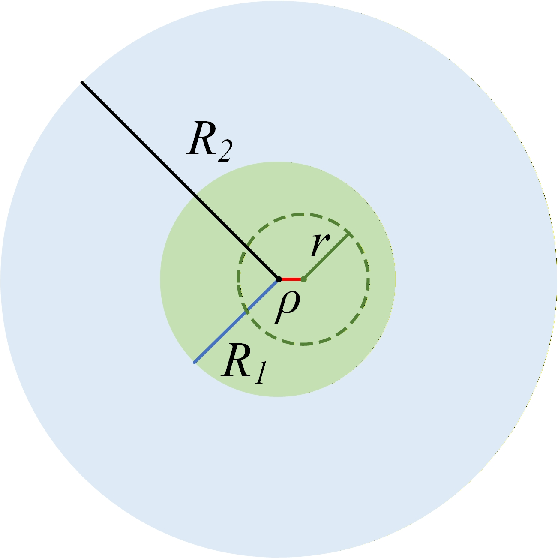}  
& \includegraphics[width=1.3in]{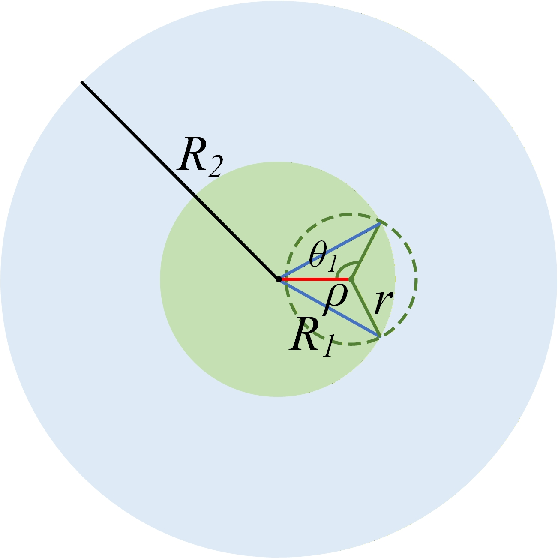}   
& \includegraphics[width=1.3in]{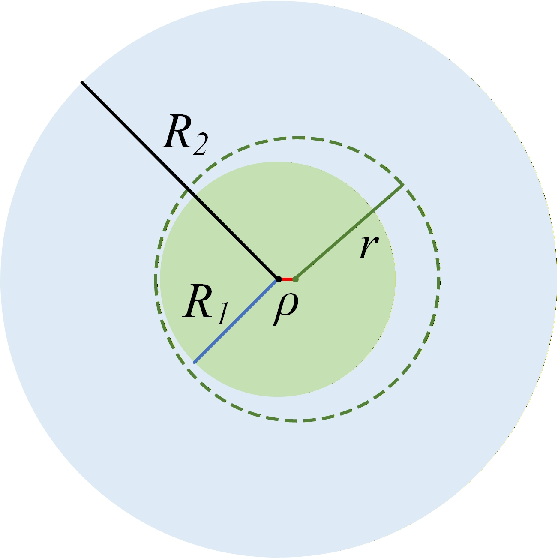}  
& \includegraphics[width=1.3in]{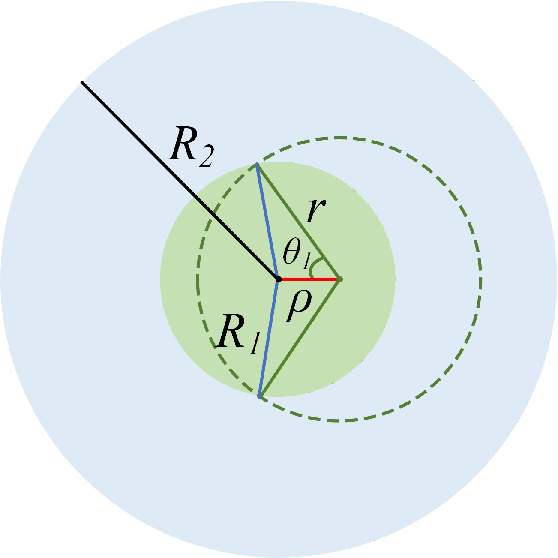}   
& \includegraphics[width=1.3in]{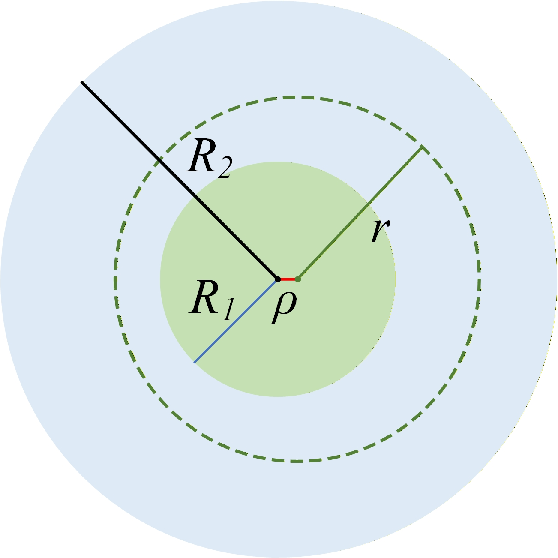}  
& \includegraphics[width=1.3in]{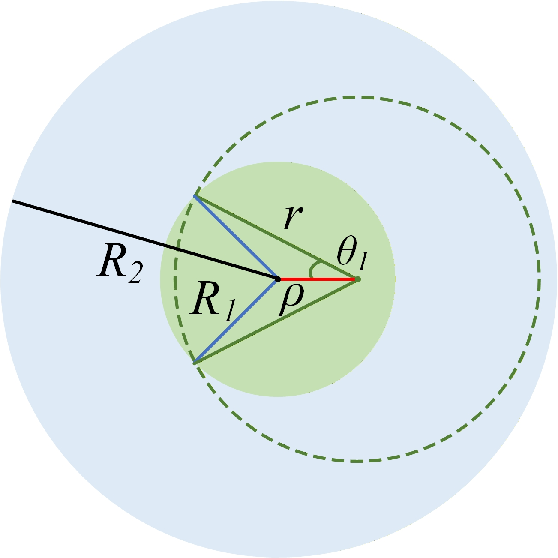}   
& \includegraphics[width=1.3in]{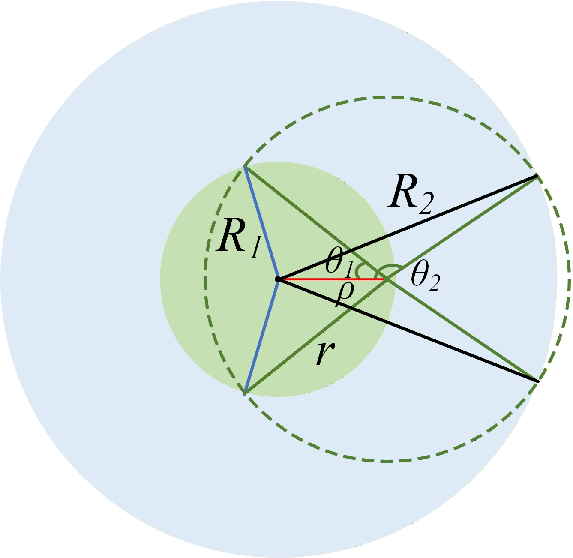} \\

\multicolumn{2}{c}{\large 
$f_{\mathbf{r}}(r)=\hspace{-8pt}\int\limits_{R_{1}-r}^{R_{1}} g_{1}(\rho,r)f_{\rho}(\rho)\,\mathrm{d}\rho$} 
& \multicolumn{2}{c}{\large 
$f_{\mathbf{r}}(r)=\hspace{-8pt}\int\limits_{0}^{r-R_{1}}\hspace{-8pt} g_{3}(r)f_{\rho}(\rho)\,\mathrm{d}\rho
+\hspace{-8pt}\int\limits_{r-R_{1}}^{R_{1}}\hspace{-6pt} g_{1}(\rho,r)f_{\rho}(\rho)\,\mathrm{d}\rho$} 
& \multicolumn{3}{c}{\large 
$f_{\mathbf{r}}(r)=\int\limits_{0}^{r-R_{1}}\hspace{-8pt} g_{3}(r)f_{\rho}(\rho)\,\mathrm{d}\rho
+\hspace{-8pt}\int\limits_{r-R_{1}}^{R_{2}-r}\hspace{-8pt} g_{1}(\rho,r)f_{\rho}(\rho)\,\mathrm{d}\rho
+\hspace{-8pt}\int\limits_{R_{2}-r}^{R_{1}}\hspace{-8pt} g_{2}(\rho,r)f_{\rho}(\rho)\,\mathrm{d}\rho$} \\

\hline \hline

\multicolumn{3}{c}{\large \textbf{Interval IV}: $\frac{R_{1}+R_{2}}{2}<r\leq2R_{1}$} 
& \multicolumn{2}{c}{\large \textbf{Interval V}: $2R_{1}<r\leq R_{2}$} 
& \multicolumn{2}{c}{\large \textbf{Interval VI}: $R_{2}<r\leq R_{1}+R_{2}$} \\ 

\large $0\leq\rho\leq R_{2}-r$ 
& \large $R_{2}-r<\rho\leq r-R_{1}$ 
& \large $r-R_{1}<\rho\leq R_{1}$ 
& 
\large $0\leq\rho\leq R_{2}-r$ 
& \large $R_{2}-r<\rho\leq R_{1}$ 
& 
\large $0\leq\rho\leq r-R_{2}$ 
& \large $r-R_{2}<\rho\leq R_{1}$ \\    

\includegraphics[width=1.3in]{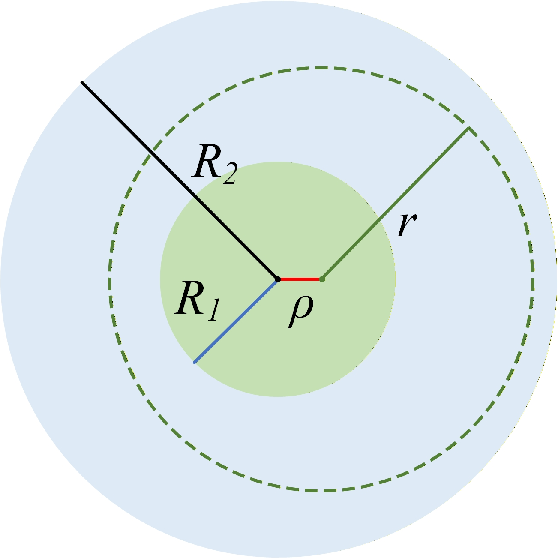}   
& \includegraphics[width=1.3in]{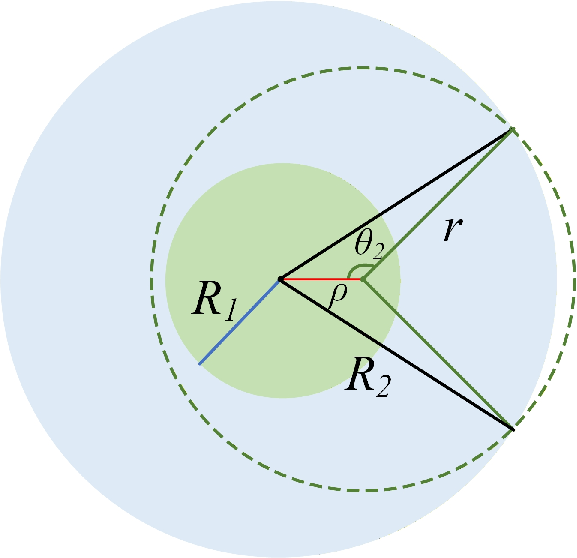}  
& \includegraphics[width=1.3in]{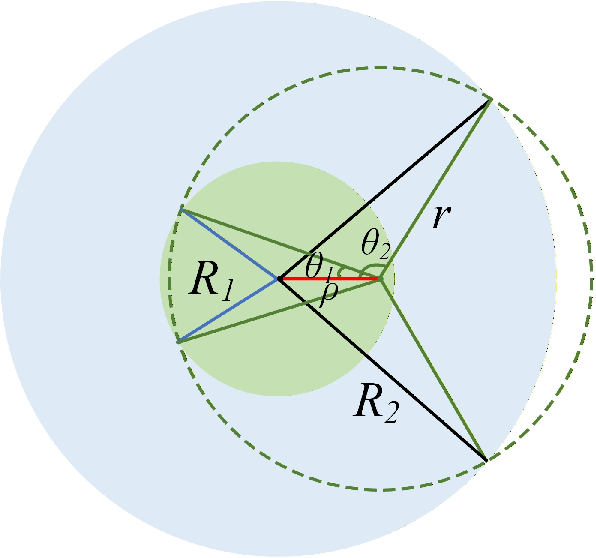}  
& \includegraphics[width=1.3in]{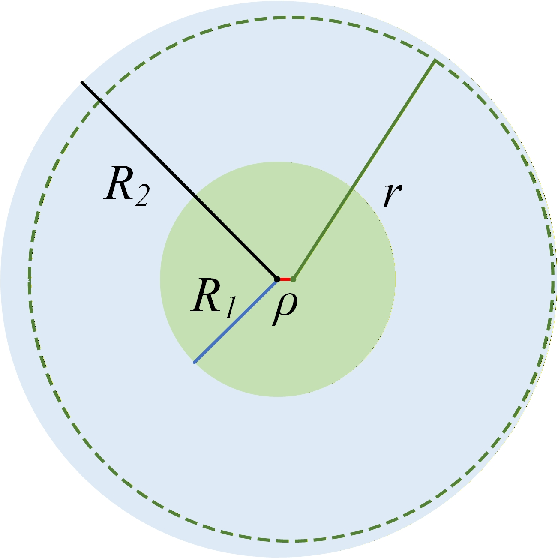}  
& \includegraphics[width=1.3in]{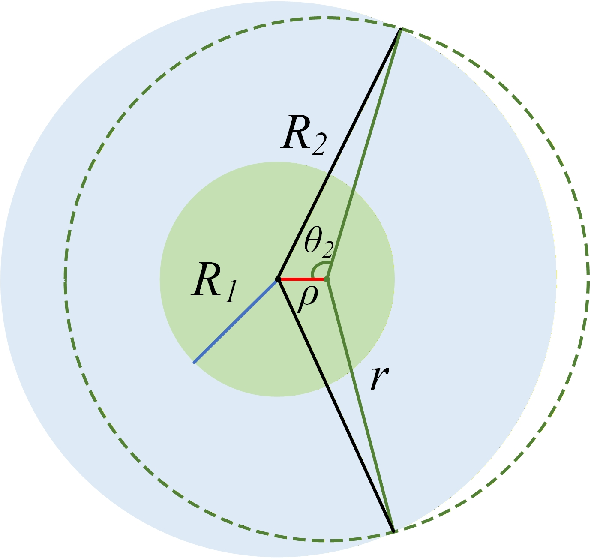}   
& \includegraphics[width=1.3in]{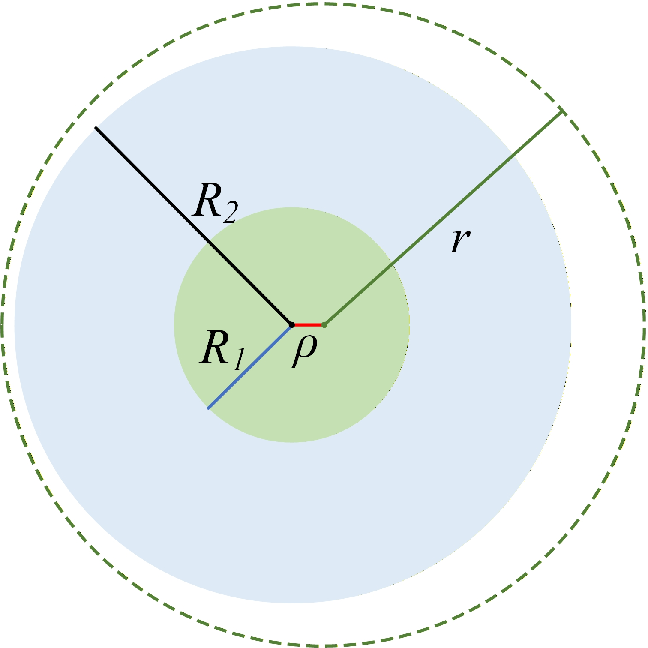} 
& \includegraphics[width=1.3in]{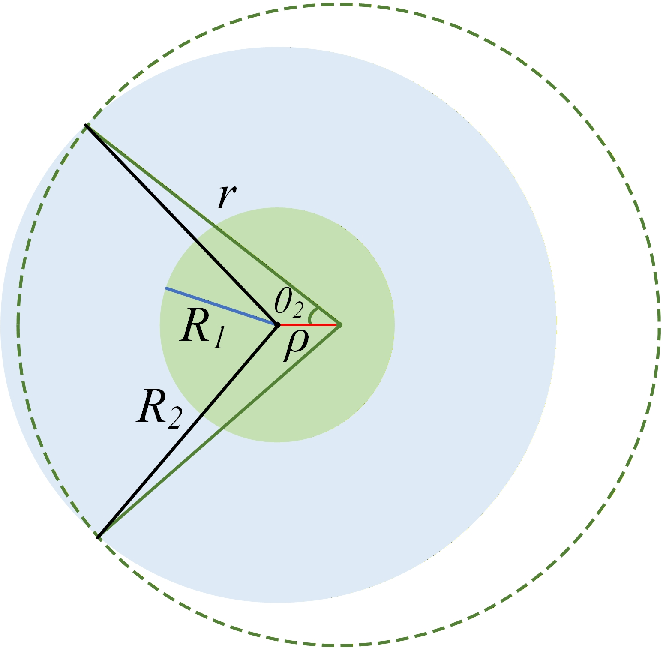} \\

\multicolumn{3}{c}{\large 
$f_{\mathbf{r}}(r)\hspace{-3pt}=\hspace{-8pt}\int\limits_{0}^{R_{2}-r}\hspace{-8pt}
 g_{3}(r)f_{\rho}(\rho)\,\mathrm{d}\rho
+\hspace{-8pt}\int\limits_{R_{2}-r}^{r-R_{1}}\hspace{-8pt} g_{4}(\rho,r)f_{\rho}(\rho)\,\mathrm{d}\rho
+\hspace{-8pt}\int\limits_{r-R_{1}}^{R_{1}}\hspace{-8pt} g_{2}(\rho,r)f_{\rho}(\rho)\,\mathrm{d}\rho$}
& \multicolumn{2}{c}{\large 
$f_{\mathbf{r}}(r)=\hspace{-8pt}\int\limits_{0}^{R_{2}-r}\hspace{-8pt} g_{3}(r)f_{\rho}(\rho)\,\mathrm{d}\rho
+\hspace{-8pt}\int\limits_{R_{2}-r}^{R_{1}}\hspace{-6pt} g_{4}(\rho,r)f_{\rho}(\rho)\,\mathrm{d}\rho$} 
& \multicolumn{2}{c}{\large 
$f_{\mathbf{r}}(r)=\hspace{-8pt}\int\limits_{r-R_{2}}^{R_{1}} g_{4}(\rho,r)f_{\rho}(\rho)\,\mathrm{d}\rho$} \\

\hline
\end{NiceTabular}
}
\end{table*}

\begin{table*}[th!]
\centering
\caption{Regime 3: $R_2>3R_1$.}
\label{Table3} 
\resizebox{0.9\textwidth}{!}{%
\renewcommand{\arraystretch}{1.0}
\begin{NiceTabular}{c c c c c}[hvlines]  
    \hline 
    \multicolumn{2}{c} {\small \textbf{Interval I}: $0\leq r\leq R_{1}$} 
    & \multicolumn{2}{c} {\small \textbf{Interval II}: $R_{1}<r\leq2R_{1}$} 
    &  {\small \textbf{Interval III}: $2R_{1}<r\leq R_{2}-R_{1}$}  \\ 

    \small $0\leq\rho\leq R_{1}-r$ 
     & \small $R_{1}-r<\rho\leq R_{1}$  & \small $0\leq\rho\leq r-R_{1}$  
    & \small $r-R_{1}<\rho\leq R_{1}$  & \small $0\leq\rho\leq R_{1}$\\   

    \includegraphics[keepaspectratio,width=1.3in]{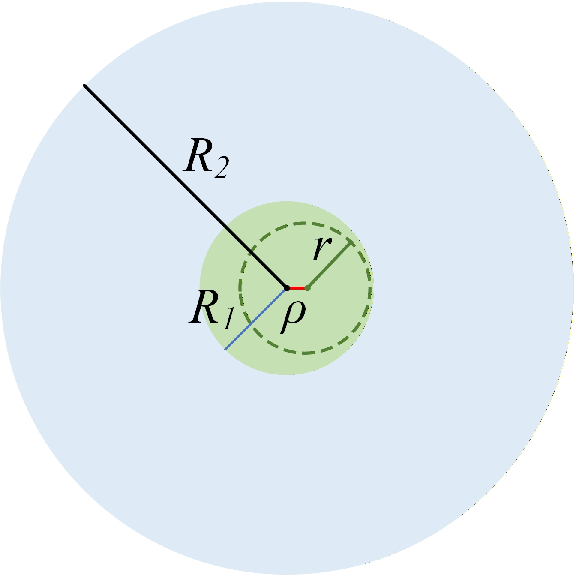}  
    & \includegraphics[keepaspectratio,width=1.3in]{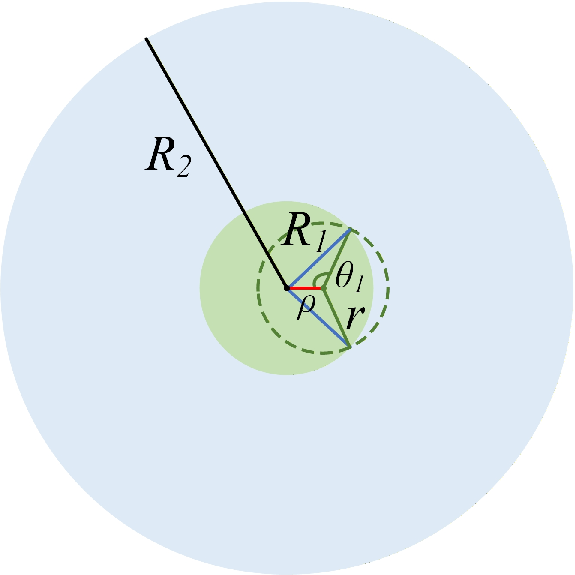}   
    &  \includegraphics[keepaspectratio,width=1.3in]{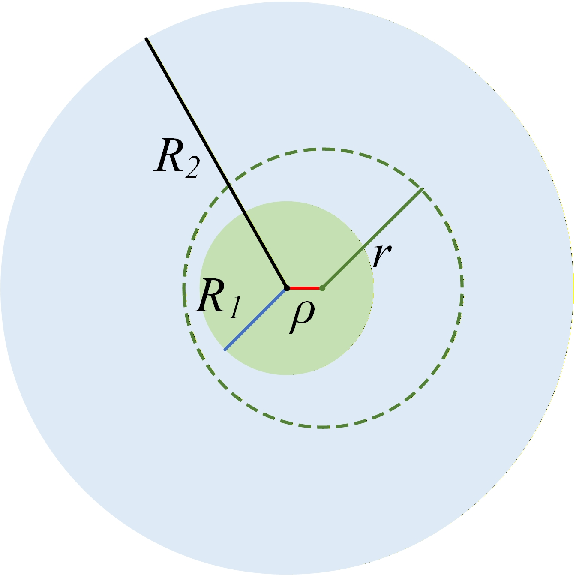}  
    & \includegraphics[keepaspectratio,width=1.3in]{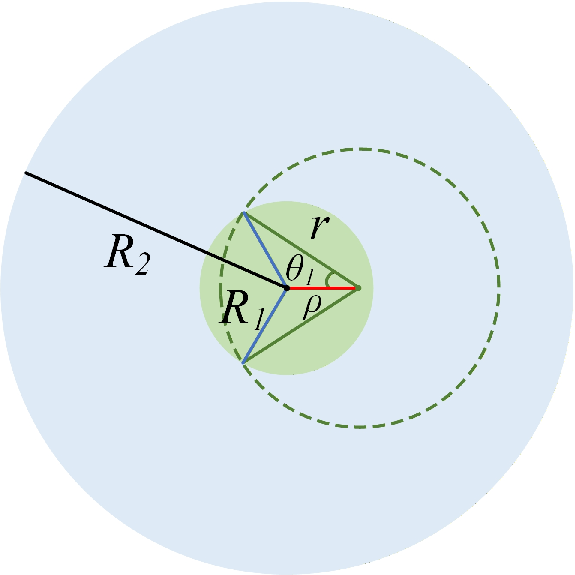}   
    & \includegraphics[keepaspectratio,width=1.3in]{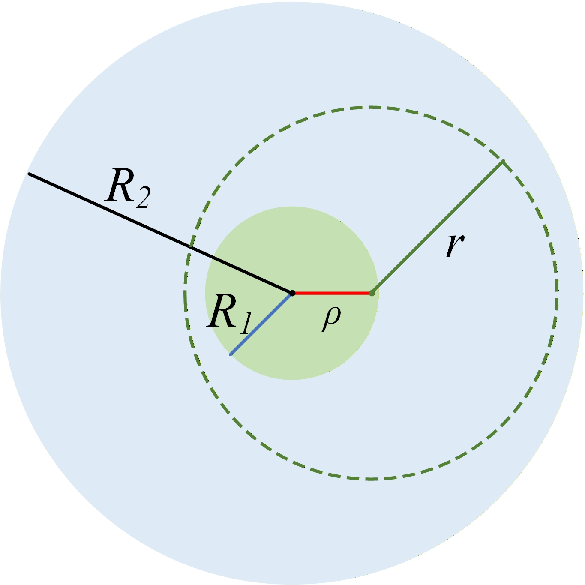}    \\
   \multicolumn{2}{c} {\small $f_{\mathbf{r}}(r)=\hspace{-8pt}\int\limits_{R_{1}-r}^{R_{1}}\hspace{-6pt} g_{1}(\rho,r)f_{\rho}(\rho) \,\mathrm{d}\rho$} &  \multicolumn{2}{c} {\small $f_{\mathbf{r}}(r)=\hspace{-8pt}\int\limits_{0}^{r-R_{1}}\hspace{-8pt}g_{3}(r)f_{\rho}(\rho) \,\mathrm{d}\rho+\hspace{-8pt}\int\limits_{r-R_{1}}^{R_{1}}\hspace{-6pt} g_{1}(\rho,r)f_{\rho}(\rho) \,\mathrm{d}\rho$}  & \small $f_{\mathbf{r}}(r)=\hspace{-3pt}\int\limits_{0}^{R_{1}}\hspace{-5pt}g_{3}(r)f_{\rho}(\rho) \,\mathrm{d}\rho$\\
\hline \hline

    \multicolumn{2}{c} {\small \textbf{Interval IV}: $R_{2}-R_{1}<r\leq R_{2}$} 
    & \multicolumn{2}{c} {\small \textbf{Interval V}: $R_{2}<r\leq R_{1}+R_{2}$}  \\ 
 \small $0\leq\rho\leq R_{2}-r$ 
     & \small $R_{2}-r<\rho\leq R_{1}$  & \small $0\leq\rho\leq r-R_{2}$  & \small $r-R_{2}<\rho\leq R_{1}$ &  \\    
    
 \includegraphics[keepaspectratio,width=1.3in]{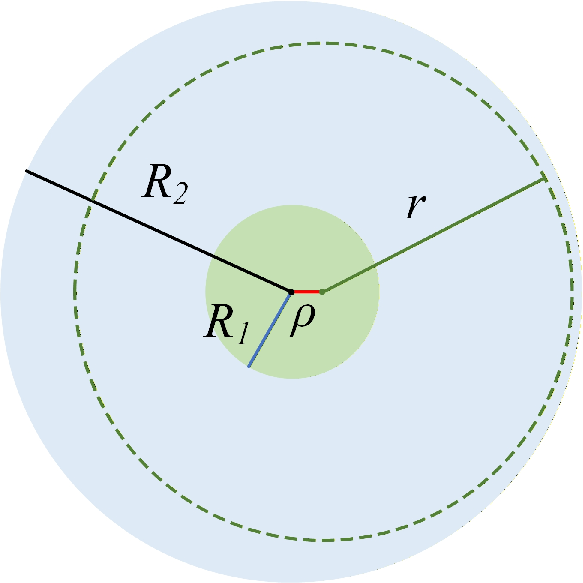}  
    & \includegraphics[keepaspectratio,width=1.3in]{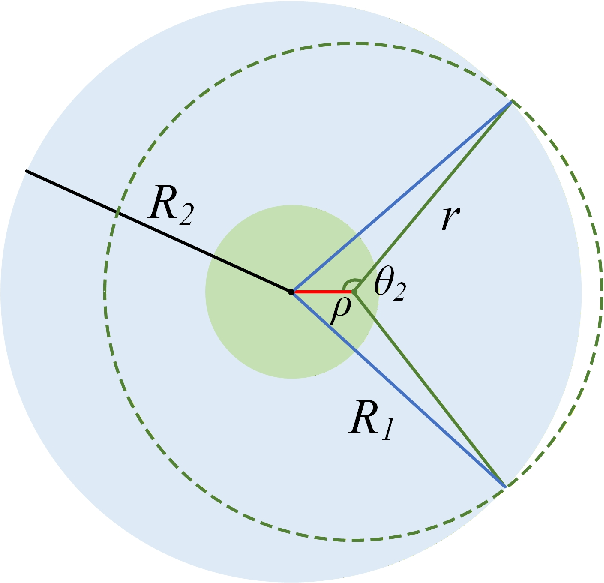}   
    &  \includegraphics[keepaspectratio,width=1.3in]{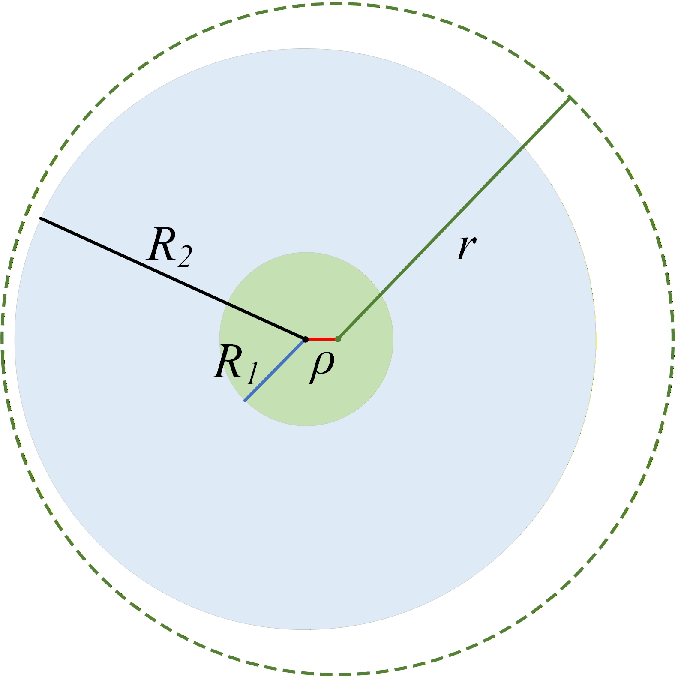}  
    & \includegraphics[keepaspectratio,width=1.3in]{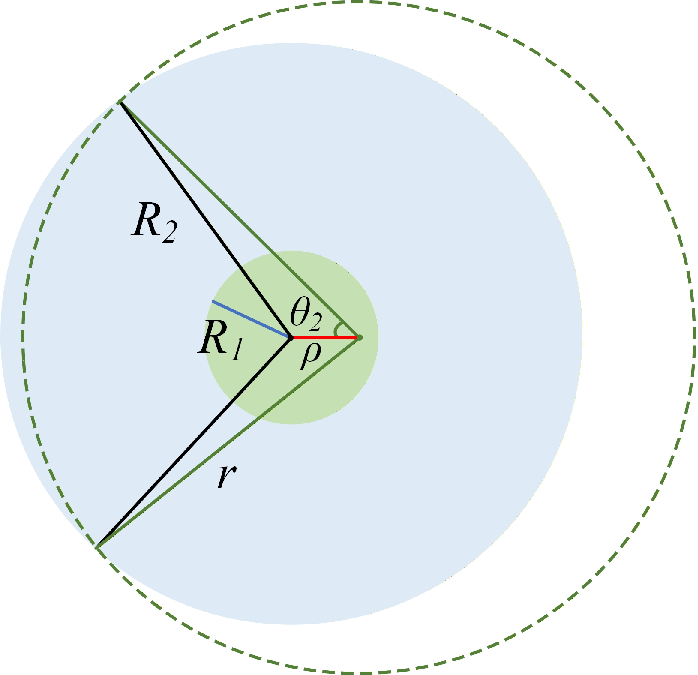}   
    &    
        \\
        \multicolumn{2}{c} {\small $f_{\mathbf{r}}(r)=\hspace{-8pt}\int\limits_{0}^{R_{2}-r}\hspace{-8pt} g_{3}(r)f_{\rho}(\rho) \,\mathrm{d}\rho+\hspace{-8pt}\int\limits_{R_{2}-r}^{R_{1}}\hspace{-6pt} g_{4}(\rho,r)f_{\rho}(\rho) \,\mathrm{d}\rho$} &  \multicolumn{2}{c} {\small $f_{\mathbf{r}}(r)=\hspace{-8pt}\int\limits_{r-R_{2}}^{R_{1}}\hspace{-5pt} g_{4}(\rho,r)f_{\rho}(\rho) \,\mathrm{d}\rho$}\\
\hline 
\end{NiceTabular}
}
\end{table*}

\section{Internodal Distance Distributions}
This section derives the internodal distance distributions for 2-D and 3-D networks via a set of theorems.

\begin{table}[t]
\centering
\scriptsize
\setlength{\tabcolsep}{5pt}
\renewcommand{\arraystretch}{1.4}
\caption{Coefficients for 2-D geometry}
\label{2DTable}
\begin{tabular}{@{}c@{\hspace{4pt}} c c c}
\toprule
 & \textbf{Coeff.} & \textbf{$s_1$} & \textbf{$s_2$} \\
\midrule
\multirow{8}{*}{\rotatebox{90}{$\boldsymbol{R_2 \leq 3R_1}$}}
 & $q_1(r)$           & \cellcolor{rowalt}$r$       & \cellcolor{rowalt}$-\frac{r(2R_1^2+r^2)}{2R_1^2}$ \\[4pt]
 & $q_2(r)$           & $-2R_1$                     & $-\frac{2(R_1^2-r^2)}{R_1}$ \\[4pt]
 & $q_3(r)$           & \cellcolor{rowalt}$rR_1$    & \cellcolor{rowalt}$-\frac{r(2R_1^2+r^2)}{2R_1}$ \\[4pt]
 & $q_4(r)$           & $-2R_1^2$                   & $-2(R_1^2-r^2)$ \\[4pt]
 & $q_5(r),q_9(r)$    & \cellcolor{rowalt}$-rR_2$   & \cellcolor{rowalt}$0$ \\[4pt]
 & $q_6(r),q_{10}(r)$ & $R_2^2$                     & $-\frac{R_2^2(2r^2-2R_1^2+R_2^2)}{R_1^2}$ \\[4pt]
 & $q_7(r),q_{11}(r)$ & \cellcolor{rowalt}$R_1^2$   & \cellcolor{rowalt}$R_1^2$ \\[4pt]
 & $q_8(r),q_{12}(r)$ & $0$                         & $\frac{r^2-3R_1^2+5R_2^2}{4R_1^2}$ \\[2pt]
\specialrule{0.8pt}{3pt}{3pt}
\multirow{6}{*}{\rotatebox{90}{$\boldsymbol{R_2 > 3R_1}$}}
 & $p_1(r)$  & \cellcolor{rowalt}$r$      & \cellcolor{rowalt}$-\frac{r(2R_1^2+r^2)}{2R_1^2}$ \\[4pt]
 & $p_2(r)$  & $-2R_1$                    & $-\frac{2(R_1^2-r^2)}{R_1}$ \\[4pt]
 & $p_3(r)$  & \cellcolor{rowalt}$-rR_2$  & \cellcolor{rowalt}$0$ \\[4pt]
 & $p_4(r)$  & $R_2^2$                    & $\frac{R_2^2(2R_1^2-2r^2-R_2^2)}{R_1^2}$ \\[4pt]
 & $p_5(r)$  & \cellcolor{rowalt}$R_1^2$  & \cellcolor{rowalt}$R_1^2$ \\[4pt]
 & $p_6(r)$  & $0$                        & $\frac{r^2-3R_1^2+5R_2^2}{4R_1^2}$ \\[2pt]
\bottomrule
\end{tabular}
\end{table}

\subsection{2-D networks}
\begin{theorem}
\label{pdf2Da} 
The PDF expression for the internodal distance distribution, $r$, in a 2-D network with $R_{2}\leq 3R_{1}$, for both $s_1$ and $s_2$, is given by
{\small
\begin{equation}
\label{Gen2D}
\hspace{-2pt}f_{\mathbf{r}}(r)\hspace{-3pt}=\hspace{-4pt}
\begin{cases}
\frac{2 r} {\mathrm{S}R_1}\Bigl(\displaystyle\pi R_1\hspace{-3pt}+\hspace{-3pt}q_1(r) \sin(\phi_0(r))+\\
q_2(r)\phi_0(r)\Bigr), & \hspace{-8pt}0\le r\le R_2-R_1,\\
\frac{2 r} {\mathrm{S}R_1^2}\hspace{-2pt}\Bigl(\hspace{-2pt}\displaystyle
q_3(r) \sin(\phi_0(r))\hspace{-3pt}+\hspace{-3pt}q_4(r)\phi_0(r)+\\ q_5(r) \sin(\phi_1(r))\hspace{-3pt}+\hspace{-3pt}q_6(r) \phi_1(r)+\\
q_7(r)\phi_2(r)+q_8(r)k(r) \Bigr), & \hspace{-8pt}R_2-R_1 <  r \leq 2 R_1,\\
\frac{2 r} {\mathrm{S}R_1^2}\hspace{-2pt}\Bigl(\hspace{-2pt}\displaystyle
q_9(r) \sin(\phi_1(r))\hspace{-3pt}+\hspace{-3pt}q_{10}(r) \phi_1(r)+\hspace{-3pt}\\ q_{11}(r) \phi_2(r) 
+ q_{12}(r) k(r) \Bigr), & \hspace{-8pt}2 R_1 < r \le R_1\hspace{-3pt}+\hspace{-3pt}R_2.
\end{cases}
\end{equation}
}where $\phi_0(r)\hspace{-3pt}=\hspace{-3pt}\theta_1(R_1,r)$, $\phi_1(r)\hspace{-3pt}=\hspace{-3pt}\theta_1(R_2,r)$, $\phi_2(r)\hspace{-3pt}=\hspace{-3pt}\theta_2(R_1,r)$, $\theta_1(\rho,r)\hspace{-3pt}=\hspace{-3pt}\arccos\left({\frac{\rho^2+r^2-R_1^2}{2r\rho}}\right)$,   $\theta_2(\rho,r)\hspace{-3pt}=\hspace{-3pt}\arccos\left({\frac{\rho^2+r^2-R_2^2}{2r\rho}}\right)$, $k(r) = \sqrt{(r^2 - (R_2 - R_1)^2)((R_1 + R_2)^2 - r^2)}$ and $\mathrm{S}\hspace{-3pt}=\hspace{-3pt}\pi(R_2^2\hspace{-3pt}-\hspace{-3pt}R_1^2)$ are system dependent parameters and $q_{1}(r)-q_{12}(r)$ are tabulated in Table \ref{2DTable}.
\end{theorem}
\begin{IEEEproof} 
The proof is provided in Appendix \ref{appA}.
\end{IEEEproof}
\begin{corollary}
\label{pdf2Db} The PDFs of $r$ for $R_{2}=R_{1}=R$ in $s_1$ and $s_2$ are given, respectively, by \cite{J:Vaiopoulos} 
{\small
\begin{equation}
f_{\mathbf{r}}(r)=\dfrac{2r}{\pi R^{2}}\arccos\hspace{-2pt}\left(\frac{r}{2R}\right), \quad 0\leq r\leq 2R ,\footnote{
See Table II with $f_r(r) \equiv g_2(\rho, r)$, $\rho \equiv R_2 \equiv R$, for s1, s4 in \cite{J:Vaiopoulos}.}
\label{pdfr1}
\end{equation}
}
{\small
\begin{equation}
f_{\mathbf{r}}(r)\hspace{-2pt}=\hspace{-2pt}\frac{4{r}^2\hspace{-2pt}\left(\hspace{-2pt}\sqrt{4{R}^2-r^{2}}\hspace{-2pt}-\hspace{-2pt}r\arccos\hspace{-2pt}\left(\frac{r}{2R}\right)\right)}{\pi R^{4}}, 0\leq r\leq2R ,\footnote{
See Table II with $f_r(r) \equiv g_2(\rho, r)$, $\rho \equiv R_2 \equiv R$, for s2, s3 in \cite{J:Vaiopoulos}.}
\label{pdfr2}
\end{equation}
}
\end{corollary}
\begin{IEEEproof}
Applying L’Hôpital’s rule to the second branch of $\eqref{Gen2D}$ for $s_1$ and $s_2$, with $R_2 = R_1 = R$, yields the inter-nodal distance distribution between a node uniformly distributed over a circular disk of radius $R$ (or moving within the disk according to the RWP model) and a second node uniformly distributed on the circumference of a circle of radius $R$.
\end{IEEEproof}
\begin{theorem}
\label{theorem1} The PDF expression for the internodal distance distribution, $r$, in a 2-D network with $R_{2}>3R_{1}$, for both $s_1$ and $s_2$, is given by
{\small
\begin{equation}
\label{Gen2Db}
\hspace{-2pt}f_{\mathbf{r}}(r)\hspace{-3pt}=\hspace{-4pt}
\begin{cases}
\frac{2 r} {\mathrm{S}R_1}\Bigl(\displaystyle\pi R_1\hspace{-3pt}+\hspace{-3pt}p_1(r) \sin(\phi_0(r))+\\
p_2(r) \phi_0(r)\Bigr),& 0\le r\le 2 R_1\\
\frac{2 r} {R_2^2-R_1^2},& 2R_1 < r \le R_2\hspace{-3pt}-\hspace{-3pt}R_1\\
\frac{2 r} {\mathrm{S}R_1^2} \,\Bigl(\displaystyle
p_3(r) \sin(\phi_1(r))\hspace{-3pt}+\hspace{-3pt}p_4(r) \phi_1(r)+\\
 p_5(r) \phi_2(r)\hspace{-3pt}+\hspace{-3pt}p_6(r) k(r)\Bigr), 
& R_2\hspace{-3pt} -\hspace{-3pt} R_1\hspace{-3pt} <\hspace{-3pt} r\hspace{-3pt} \le\hspace{-3pt} R_1\hspace{-3pt} +\hspace{-3pt} R_2.
\end{cases}
\end{equation}
}where $p_{1}(r)-p_6(r)$ are tabulated in Table \ref{2DTable}.
\end{theorem}
\begin{IEEEproof} 
The proof is provided in Appendix \ref{appA}.
\end{IEEEproof}
\begin{corollary}
\label{pdf2Dc} The PDF of $r$ for $R_{1}=0$ and  $R_{2}=R$ is given by $f_{\mathbf{r}}(r)=2r/R^{2}$, with $0\leq r\leq R$.
\end{corollary}
\begin{IEEEproof}
The PDF is obtained by setting $R_1 = 0$ and $R_2 = R$ in the second branch of $\eqref{Gen2Db}$ for both $s_1$ and $s_2$, describing the distance distribution of a node uniformly distributed over a circular disk of radius $R$.
\end{IEEEproof}
\subsection{3-D networks}
\begin{table}[ht]
\centering
\caption{Coefficients for 3-D geometry}
\label{3DTable}
\resizebox{\columnwidth}{!}{%
\renewcommand{\arraystretch}{1.5}
\begin{tabular}{@{}c@{\hspace{4pt}} c c c}
\toprule
 & \textbf{Coeff.} & \textbf{$s_1$} & \textbf{$s_2$} \\
\midrule
\multirow{16}{*}{\rotatebox{90}{\shortstack{$\mathbf{R_2 \leq 3R_1}$ and $\mathbf{R_2 > 3R_1}$}}}
 & $a_3$ & \cellcolor{rowalt}$-\frac{9}{4R_1(R_1^3-R_2^3)}$ & \cellcolor{rowalt}$0$ \\[3pt]
 & $a_4$ & $0$ & $-\frac{35}{18R_1^2(R_1^3-R_2^3)}$ \\[3pt]
 & $a_5$ & \cellcolor{rowalt}$\frac{3}{16R_1^3(R_1^3-R_2^3)}$ & \cellcolor{rowalt}$-\frac{35}{16R_1^3(R_1^3-R_2^3)}$ \\[3pt]
 & $a_6$ & $0$ & $\frac{455}{144R_1^4(R_1^3-R_2^3)}$ \\[3pt]
 & $a_7$ & \cellcolor{rowalt}$0$ & \cellcolor{rowalt}$-\frac{287}{288R_1^5(R_1^3-R_2^3)}$ \\[3pt]
 & $a_8$ & $0$ & $0$ \\[3pt]
 & $a_9$ & \cellcolor{rowalt}$0$ & \cellcolor{rowalt}$\frac{65}{2304R_1^7(R_1^3-R_2^3)}$ \\[3pt]
\cmidrule{2-4}
 & $c_1$ & \cellcolor{rowalt}$\frac{9(R_1-R_2)(R_1+R_2)^2}{16R_1^3(R_1^2+R_1R_2+R_2^2)}$
 & \cellcolor{rowalt}$\frac{35(R_1-R_2)^2(R_1+R_2)^3(29R_1^2-13R_2^2)}{2304R_1^7(R_1^2+R_1R_2+R_2^2)}$ \\[6pt]
 & $c_2$ & $-\frac{3(R_1^3+R_2^3)}{2R_1^3(R_1^3-R_2^3)}$
 & $-\frac{72R_1^7+245R_1^4R_2^3-238R_1^2R_2^5+65R_2^7}{48R_1^7(R_1^3-R_2^3)}$ \\[6pt]
 & $c_3$ & \cellcolor{rowalt}$\frac{9(R_1^2+R_2^2)}{8R_1^3(R_1^3-R_2^3)}$
 & \cellcolor{rowalt}$\frac{35(R_1+R_2)(25R_1^4+88R_1^2R_2^2-65R_2^4)}{576R_1^7(R_1^2+R_1R_2+R_2^2)}$ \\[6pt]
 & $c_4$ & $0$
 & $\frac{35(17R_1^2R_2^3-13R_2^5)}{72R_1^7(R_1^3-R_2^3)}$ \\[6pt]
 & $c_5$ & \cellcolor{rowalt}$-\frac{3}{16R_1^3(R_1^3-R_2^3)}$
 & \cellcolor{rowalt}$-\frac{35(7R_1^4+34R_1^2R_2^2-65R_2^4)}{384R_1^7(R_1^3-R_2^3)}$ \\[6pt]
 & $c_6$ & $0$ & $-\frac{455R_2^3}{144R_1^7(R_1^3-R_2^3)}$ \\[3pt]
 & $c_7$ & \cellcolor{rowalt}$0$ & \cellcolor{rowalt}$\frac{7(17R_1^2+65R_2^2)}{576R_1^7(R_1^3-R_2^3)}$ \\[3pt]
 & $c_8$ & $0$ & $0$ \\[3pt]
 & $c_9$ & \cellcolor{rowalt}$0$ & \cellcolor{rowalt}$-\frac{65}{2304R_1^7(R_1^3-R_2^3)}$ \\[3pt]
\specialrule{0.8pt}{3pt}{3pt}
\multirow{7}{*}{\rotatebox{90}{only $\mathbf{R_2 \leq 3R_1}$}}
 & $b_1$ & \cellcolor{rowalt}$\frac{9(R_1-R_2)(R_1+R_2)^2}{16R_1^3(R_1^2+R_1R_2+R_2^2)}$
 & \cellcolor{rowalt}$\frac{35(R_1-R_2)^2(R_1+R_2)^3(29R_1^2-13R_2^2)}{2304R_1^7(R_1^2+R_1R_2+R_2^2)}$ \\[6pt]
 & $b_2$ & $\frac{3}{2R_1^3}$
 & $\frac{(R_1-R_2)(72R_1^5+144R_1^4R_2+216R_1^3R_2^2+43R_1^2R_2^3-130R_1R_2^4-65R_2^5)}{48R_1^7(R_1^2+R_1R_2+R_2^2)}$ \\[6pt]
 & $b_3$ & \cellcolor{rowalt}$-\frac{9(R_1+R_2)}{8R_1^3(R_1^2+R_1R_2+R_2^2)}$
 & \cellcolor{rowalt}$\frac{35(R_1+R_2)(25R_1^4+88R_1^2R_2^2-65R_2^4)}{576R_1^7(R_1^2+R_1R_2+R_2^2)}$ \\[6pt]
 & $b_4$ & $0$
 & $\frac{-35(4R_1^4+4R_1^3R_2+4R_1^2R_2^2-13R_1R_2^3-13R_2^4)}{72R_1^7(R_1^2+R_1R_2+R_2^2)}$ \\[6pt]
 & $b_5$ & \cellcolor{rowalt}$0$
 & \cellcolor{rowalt}$-\frac{35(R_1+R_2)(31R_1^2+65R_2^2)}{384R_1^7(R_1^2+R_1R_2+R_2^2)}$ \\[6pt]
 & $b_6$ & $0$ & $\frac{455}{144R_1^7}$ \\[3pt]
 & $b_7$ & \cellcolor{rowalt}$0$ & \cellcolor{rowalt}$-\frac{455(R_1+R_2)}{576R_1^7(R_1^2+R_1R_2+R_2^2)}$ \\[3pt]
\bottomrule
\end{tabular}
}
\end{table}
\begin{theorem}
\label{pdf3Da} The PDF expression for the internodal distance distribution, $r$, in a 3-D network with $R_{2}\leq 3R_{1}$, for both $s_1$ and $s_2$, is given by
\begin{equation}
\label{Gen3Da}
f_{\mathbf{r}}(r)=
\begin{cases}
\displaystyle \textstyle \sum_{n=3}^{9} a_n r^n, & 0 \le r \le R_2 - R_1,\\[2pt]
\displaystyle \textstyle \sum_{n=1}^{7} b_n r^n, & R_2 - R_1 < r \le 2 R_1,\\[2pt]
\displaystyle \textstyle \sum_{n=1}^{9} c_n r^n, & 2 R_1 < r \le R_1 + R_2.
\end{cases}
\end{equation}
where $a_n$, $b_n$, and $c_n$ are tabulated in Table \ref{3DTable}.
\end{theorem}

\begin{IEEEproof} 
The proof is provided in Appendix \ref{appB}.
\end{IEEEproof}
\begin{corollary}
\label{Corol1} The PDFs of $r$ for $R_{2}=R_{1}=R$ in $s_1$ and $s_2$ are given, respectively, by
{\small
\begin{equation}
f_{\mathbf{r}}(r)=\dfrac{3{r}^2(2R-r)}{4R^{4}}, \quad 0\leq r \leq2R ,\footnote{
See Table III with $f_r(r) \equiv g_2(\rho, r)$, $\rho \equiv R_2 \equiv R$, for s1, s4 in \cite{J:Vaiopoulos}.}
\label{pdfr3}
\end{equation}
}
{\small
\begin{equation}
f_{\mathbf{r}}(r)\hspace{-2pt}=\hspace{-2pt}\dfrac{35{r}^3{(2R-r)}^2(25{R}^2-13{(r-R)}^2)}{864R^{8}}, 0\leq r\leq2R ,\footnote{
See Table III with $f_r(r) \equiv g_2(\rho, r)$, $\rho \equiv R_2 \equiv R$, for s2, s3 in \cite{J:Vaiopoulos}.}
\label{pdfr4}
\end{equation}
}
\end{corollary}
\begin{IEEEproof}
The PDF follows by setting $R_2 = R_1 = R$ in the second branch of $\eqref{Gen3Da}$ for $s_1$ and $s_2$, and describes the inter-nodal distance between a node uniformly distributed within the spherical region (or moving according to the RWP model) and a second node uniformly distributed on the surface of a sphere of radius $R$.
\end{IEEEproof}
\begin{theorem}
\label{pdf3Db} The PDF expression for the internodal distance distribution, $r$, in a 3-D network with $R_{2}>3R_{1}$, for both $s_1$ and $s_2$, is given by
\begin{equation}
\label{Gen3Db}
f_{\mathbf{r}}(r)=
\begin{cases}
\displaystyle \textstyle \sum_{n=3}^{9} a_n r^n, 
& 0 \le r \le 2 R_1 \\
 \frac{3 r^2}{R_2^3-R_1^3},& 2 R_1 < r \le R_2- R_1\\[2pt]
\displaystyle \textstyle \sum_{n=1}^{9} c_n r^n, 
& R_2 - R_1 < r \le R_1 + R_2.
\end{cases}
\end{equation}
where $a_n$ and $c_n$ are tabulated in Table \ref{3DTable}.
\end{theorem}
\begin{IEEEproof} 
The proof is provided in Appendix \ref{appB}.
\end{IEEEproof}
\begin{corollary}
\label{Corol4} The PDF of $r$ for $R_{1}=0$ and  $R_{2}=R$ is given by $f_{\mathbf{r}}(r)=3{r}^2/R^{3}$, with $0\leq r\leq R$.
\end{corollary}
\begin{IEEEproof}
The PDF is obtained by setting $R_1 = 0$ and $R_2 = R$ in the second branch of $\eqref{Gen3Db}$ for both $s_1$ and $s_2$, describing the distance distribution of a node uniformly distributed within a spherical region of radius $R$.
\end{IEEEproof}

\subsection{Approximation with beta distributions}
\label{sec:beta}
Although exact closed-form PDFs are derived, their use in wireless communication problems can be further simplified by approximating them using a beta distribution  with PDF \cite{J:Ermolova}.
\begin{equation}
{f_{\bf{x}}}(x) = \frac{{{x^{\alpha  - 1}}{{(1 - x)}^{\beta  - 1}}}}{{B(\alpha ,\beta )}},\begin{array}{*{20}{c}}
{}&{}&{x \in \left[ {0,1} \right]},
\end{array}
\label{betaDist}
\end{equation} 
where $\alpha$ and $\beta$ are the distribution parameters, and $B(\alpha,\beta)$ is the beta function. The parameters are obtained via moment matching (first and second moments) with the target distribution, yielding the corresponding optimal values. Figure \ref{Figure4} illustrates representative PDF plots for a 2-D network, with $R_{2}/R_{1}=2$ and $R_{2}/R_{1}=3.5$, validated against Monte Carlo simulations using $10^{5}$ independent realizations. 

Figure \ref{Figure5} depicts the Kullback-Leibler ($\mathrm{KL}$) divergence between the exact PDFs 
and their beta-based approximations as a function of $R_{2}/R_{1}$.  In 2-D networks, the $\mathrm{KL} = 
10^{-2}$ threshold is reached at $R_{2}/R_{1} \approx 7$ and $R_{2}/R_{1} 
\approx 5$ for $s_1$ and $s_2$, respectively; the tighter bound in $s_2$ 
stems from the spatial bias of the RWP stationary distribution, which 
sharpens the PDF in a manner less compatible with the beta form. In 3-D, accuracy degrades more rapidly, with thresholds 
at $R_{2}/R_{1} \approx 5.5$ and $R_{2}/R_{1} \approx 3.5$
--- though the practical need for approximation is reduced 
since the exact 3-D expressions are already polynomial in 
$r$ and thus inherently tractable.

\begin{figure}[t]
\centering
\begin{minipage}{0.23\textwidth}
\centering
\includegraphics[width=\linewidth]{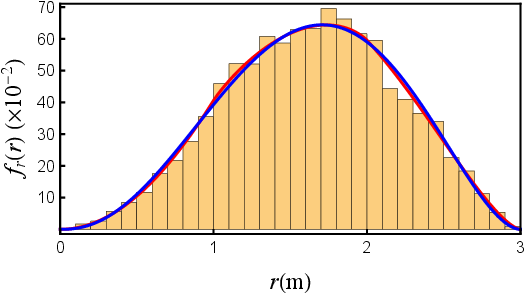}\\
{\scriptsize (a) $s_1$, $\{\alpha,\beta\}=\{3.419,2.832\}$}
\end{minipage}\hfill
\begin{minipage}{0.23\textwidth}
\centering
\includegraphics[width=\linewidth]{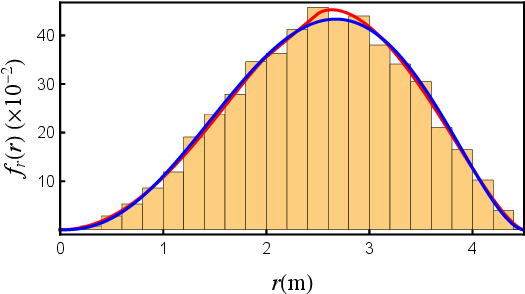}\\
{\scriptsize (b) $s_1$, $\{\alpha,\beta\}=\{3.537,2.735\}$}
\end{minipage}

\vspace{0.5em}

\begin{minipage}{0.23\textwidth}
\centering
\includegraphics[width=\linewidth]{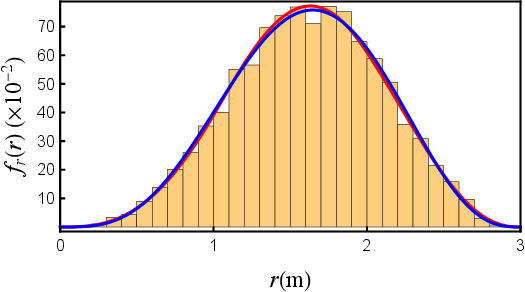}\\
{\scriptsize (c) $s_2$, $\{\alpha,\beta\}=\{4.589,3.951\}$}
\end{minipage}\hfill
\begin{minipage}{0.23\textwidth}
\centering
\includegraphics[width=\linewidth]{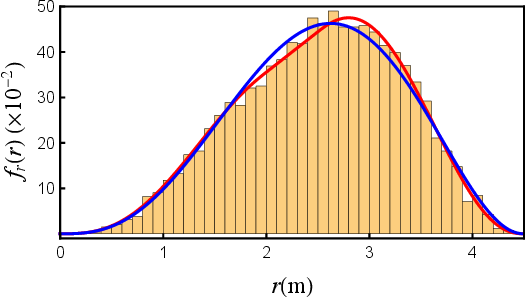}\\
{\scriptsize (d) $s_2$, $\{\alpha,\beta\}=\{3.993,3.140\}$}
\end{minipage}

\caption{PDFs for a 2-D network: $\{R_1, R_2\}=\{1,2\}\mathrm{m}$ in (a), (c) and $\{1,3.5\}\mathrm{m}$ in (b), (d); analytical (red), beta (blue), and Monte Carlo (histograms).}
\label{Figure4}
\end{figure}

\begin{figure}[ht]
\centering
\includegraphics[width=0.8\linewidth]{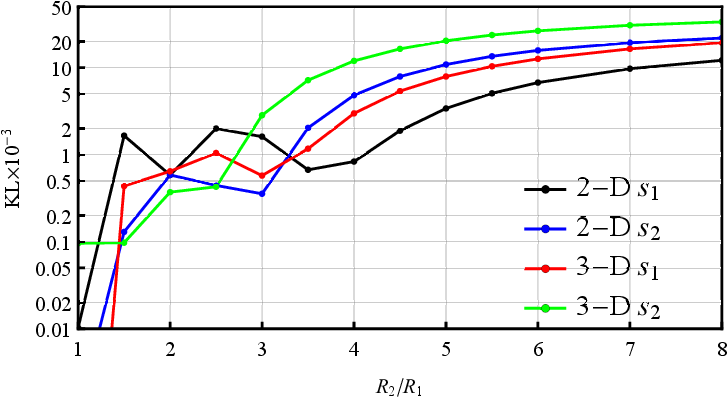}
\caption{$\mathrm{KL}$ divergence against $R_2/R_1$.}
\label{Figure5}
\end{figure}
 
\section{Conclusions}
This letter derived closed-form expressions for the internodal distance distribution in 2-D disk–annulus and 3-D sphere–shell wireless networks. Two scenarios were considered: independently distributed nodes in distinct regions and a hybrid case with a static outer node and a mobile inner node following the RWP model.  In addition, simple beta-based approximations have been shown to provide high accuracy. The results offer a useful tool for performance evaluation in bounded concentric wireless regions.

\appendices
\section{Proof of the Main Results}
\label{app}

\subsection{2-D annulus}
\label{appA}
Due to the uniform distribution of node 2 within the annular region defined by radii $R_1$ and $R_2$, the conditional PDF can be expressed as the ratio of the length of the circular arc of radius $r$ lying within the annulus to the total area of the annulus, $\mathrm{S}$. Depending on the relative position of the circle of radius $r$, the conditional PDF is given by $g_1(\rho,r)\hspace{-2pt}=\hspace{-2pt} \frac{2r}{\mathrm{S}}(\pi\hspace{-2pt} -\hspace{-2pt} \theta_1(\rho,r))$, $g_2(\rho,r)\hspace{-2pt} =\hspace{-2pt} \frac{2r}{\mathrm{S}}(\theta_2(\rho,r)\hspace{-2pt}-\hspace{-2pt}\theta_1(\rho,r))$, $g_3(r)\hspace{-2pt} =\hspace{-2pt} \frac{2\pi r}{\mathrm{S}}$, and $g_4(\rho,r)\hspace{-2pt} =\hspace{-2pt} \frac{2r}{\mathrm{S}}\theta_2(\rho,r)$.
Here, $\theta_1(\rho,r)$ and $\theta_2(\rho,r)$ represent the angular limits defined by the intersections of the circle of radius $r$ with the inner and outer boundaries of the annulus, respectively.

\subsection{3-D spherical shell}
\label{appB}
For a node uniformly distributed within a spherical shell of inner radius $R_1$ and outer radius $R_2$, the conditional PDF can be expressed as the ratio between the surface area of the spherical segment of radius $r$ contained in the shell and the total volume of the shell. Depending on the intersection regime, we obtain: $g_1(\rho,r)\hspace{-2pt}=\hspace{-2pt}\frac{2\pi r^2}{\mathrm{V}}(1\hspace{-2pt}+\hspace{-2pt}\cos(\theta_1(\rho,r)))$, $g_2(\rho,r)\hspace{-2pt}=\hspace{-2pt}\frac{2\pi r^2}{\mathrm{V}}(\cos(\theta_1(\rho,r))\hspace{-2pt}-\hspace{-2pt}\cos(\theta_2(\rho,r)))$, $g_3(r)=\frac{4\pi r^2}{\mathrm{V}}$, and $g_4(\rho,r)=\frac{2\pi r^2}{\mathrm{V}}(1\hspace{-2pt}-\hspace{-2pt}\cos(\theta_2(\rho,r)))$ where $\mathrm{V}=(4/3)\pi(R_2^3-R_1^3)$.
Here, $\theta_1(\rho,r)$ and $\theta_2(\rho,r)$ denote the angular limits determined by the intersections of the sphere of radius $r$ with the inner and outer boundaries of the spherical shell, respectively.
The corresponding conditional PDFs for the annulus or spherical shell, along with the associated interval-dependent integrals, are provided in Tables I–III.

\bibliographystyle{IEEEtran}
\bibliography{IEEEabrv,References}

\begin{thebibliography}{1}
\providecommand{\url}[1]{#1}
\csname url@samestyle\endcsname
\providecommand{\newblock}{\relax}
\providecommand{\bibinfo}[2]{#2}
\providecommand{\BIBentrySTDinterwordspacing}{\spaceskip=0pt\relax}
\providecommand{\BIBentryALTinterwordstretchfactor}{4}
\providecommand{\BIBentryALTinterwordspacing}{\spaceskip=\fontdimen2\font plus
\BIBentryALTinterwordstretchfactor\fontdimen3\font minus \fontdimen4\font\relax}
\providecommand{\BIBforeignlanguage}[2]{{%
\expandafter\ifx\csname l@#1\endcsname\relax
\typeout{** WARNING: IEEEtran.bst: No hyphenation pattern has been}%
\typeout{** loaded for the language `#1'. Using the pattern for}%
\typeout{** the default language instead.}%
\else
\language=\csname l@#1\endcsname
\fi
#2}}
\providecommand{\BIBdecl}{\relax}
\BIBdecl

\bibitem{J:Armeniakos}
H.~K. Armeniakos, P.~S. Bithas, S.~A. Tegos, A.~G. Kanatas, and G.~K. Karagiannidis, ``Stochastic geometry for modeling and analysis of sensing and communications: A survey,'' \emph{{IEEE} Commun. Surveys Tuts.}, vol.~28, pp. 2691--2724, 2026.

\bibitem{B:Mathai}
A.~M. Mathai, \emph{An Introduction to Geometrical Probability: Distributional Aspects with Applications}.\hskip 1em plus 0.5em minus 0.4em\relax Amsterdam, The Netherlands: Gordon \& Breach, 1999.

\bibitem{J:Moltchanov}
D.~Moltchanov, ``Distance distributions in random networks,'' \emph{Ad Hoc Netw.}, vol.~10, no.~6, pp. 1146--1166, 2012.

\bibitem{J:Parry}
M.~Parry and E.~Fischbach, ``Probability distribution of distance in a uniform ellipsoid: Theory and applications to physics,'' \emph{J. Math. Phys.}, vol.~41, no.~4, pp. 2417--2433, 2000.

\bibitem{J:Hyytia}
E.~Hyytia, P.~Lassila, and J.~Virtamo, ``Spatial node distribution of the random waypoint mobility model with applications,'' \emph{{IEEE} Trans. Mobile Comput.}, vol.~5, no.~6, pp. 680--694, Jun. 2006.

\bibitem{J:Vaiopoulos}
N.~Vaiopoulos, A.~Vavoulas, H.~G. Sandalidis, K.~K. Delibasis, and D.~Varoutas, ``Internodal distance distributions for static and mobile nodes in 2{D}/3{D} wireless networks,'' \emph{{IEEE} Commun. Lett.}, vol.~30, pp. 1395--1399, 2026.

\bibitem{J:Ermolova}
N.~Y. Ermolova and O.~Tirkkonen, ``Using beta distributions for modeling distances in random finite networks,'' \emph{IEEE Commun. Lett.}, vol.~20, no.~2, pp. 308--311, Feb. 2015.

\end{thebibliography}
\end{document}